\documentclass[11pt]{article}

\usepackage[T1]{fontenc}
\usepackage[utf8]{inputenc}
\usepackage{lmodern}
\usepackage{geometry}
\usepackage{microtype}
\usepackage{amsmath,amssymb,amsthm}
\usepackage{mathtools}
\usepackage{booktabs}
\usepackage{array}
\usepackage{graphicx}
\usepackage{subcaption}
\usepackage{natbib}
\usepackage{hyperref}
\usepackage{xcolor}
\usepackage{setspace}

\hypersetup{
  colorlinks=true,
  linkcolor=blue!50!black,
  citecolor=blue!50!black,
  urlcolor=blue!50!black,
  pdftitle={What Does a Benford Test Actually Test? Marginal Conformity, Sampling Structure,\\ and Forensic Inference},
  pdfauthor={Arthur Charpentier}
}

\newtheorem{proposition}{Proposition}
\newcommand{\Unif}{\operatorname{Unif}}
\newcommand{\Law}{\mathcal{L}}
\newcommand{\E}{\mathbb{E}}
\newcommand{\Pp}{\mathbb{P}}

\title{\textbf{What Does a Benford Test Actually Test?\\
Marginal Conformity, Sampling Structure,\\and Forensic Inference}}
\author{Arthur Charpentier\\
Université du Québec à Montréal (UQAM), Canada\\
Kyoto University, Japan\\[0.5em]
\small Corresponding author: \texttt{charpentier.arthur@uqam.ca}}
\date{}

\begin{document}
\maketitle

\begin{abstract}
Benford's law specifies a marginal distribution for significant digits, whereas the usual first-digit Pearson $p$-value is calibrated under an independent multinomial sampling model. We separate these statements with four constructions that share the same one-time or pooled Benford target but have different joint structures. Under a wrapped-Gaussian circular Markov construction, the nominal 5\% Pearson test rejects 21.6\% of samples at a fixed persistence level even though every one-time marginal is exactly Benford; a sequence-level second-order moment-matching calibration reduces the rejection rate to 4.8\%. The same strategy performs well in the randomized-rotation and random-composition designs. The effect persists across changes in sample size and sequence length, and a grid approximation to a continuous log-significand Cram\'er--von Mises discrepancy shows the same design dependence. In the wrapped-Gaussian design examined here, corrected procedures retain rejection probabilities that rise with the size of a smooth marginal departure. Finally, we distinguish calibration from an information boundary: significand-only methods cannot detect changes that leave the complete significand process unchanged, but they readily detect perturbations that alter it. Benford conformity is a marginal statement; a Benford $p$-value is valid only relative to a specified statistic, sampling law, and calibration procedure; an integrity claim requires substantive competing models.
\end{abstract}


\section{Introduction}
\label{sec:intro}

Benford's law is often used as a simple benchmark for numerical data. Let $D_1$ denote the first significant decimal digit of a positive observation. The first-digit law is
\begin{equation}
\Pp(D_1=d)
=
\log_{10}\left(1+\frac{1}{d}\right),
\qquad d=1,\ldots,9.
\label{eq:firstdigit}
\end{equation}
The regularity goes back to \citet{newcomb1881} and \citet{benford1938}; scale invariance and uniformity modulo one provide its modern mathematical foundation \citep{pinkham1961,diaconis1977,bergerhill2011}. Its statistical appeal is immediate: equation~\eqref{eq:firstdigit} provides a low-dimensional reference distribution without requiring a full model for magnitudes.

That convenience has made Benford screening common in accounting, taxation, scientific-data assessment, international trade, and election forensics \citep{varian1972,nigrini1996,durtschi2004,diekmann2007,deckert2011}. The discussion has a long history in \emph{The American Statistician} itself: \citet{varian1972} presented Benford screening as a diagnostic, \citet{cho2007} examined its use in campaign-finance fraud detection, and \citet{fewster2009} gave an intuitive account of when Benford behavior should arise. A substantial methodological literature now studies goodness-of-fit procedures, continuous-significand tests, and fraud-oriented alternatives \citep{lesperance2016,barabesi2018,barabesi2022,barabesi2023,barabesi2026}. Our question precedes the choice among such statistics: \emph{what sampling model is implicit in the reported $p$-value?}

The issue is classical outside the Benford setting. Pearson statistics for categorical data are not generally calibrated by a standard chi-square law under clustering, complex sampling, or serial dependence. \citet{raoscott1981} showed that survey design can replace the usual chi-square limit by a weighted sum of independent $\chi^2_1$ variables, and subsequent work developed design corrections for categorical tables \citep{raoscott1984,raoscott1987}. Related effects arise for serially dependent categorical observations \citep{altham1979,tavarealtham1983}. We therefore do not claim a new weighted-chi-square limit or a new general design correction. Our contribution is instead threefold: to give exact Benford-preserving constructions that isolate sampling structure while holding the marginal target fixed; to show, analytically and numerically, how standard design-aware ideas change the resulting Benford calibration; and to separate that calibration question from the information available for a substantive forensic comparison.

The paper distinguishes three statements:
\[
\text{Benford marginal conformity},\qquad
\text{calibration of a Benford statistic},\qquad
\text{evidence about integrity}.
\]
A Benford marginal describes one observation. A conventional Pearson $p$-value describes a statistic under an independent multinomial sampling law. An integrity claim compares substantive data-generating mechanisms. These are related only when the observation model and the competing substantive hypotheses connect them.

We proceed in three stages. First, we express the Pearson statistic under a general covariance structure and construct four processes with the same Benford target but sharply different finite-sample behavior: i.i.d. sampling, a wrapped-Gaussian process on the logarithmic circle, an irrational rotation with random phase, and a balanced latent mixture. We then show that sequence-level covariance information substantially improves calibration through Rao--Scott-type moment matching and a sequence-aware Wald statistic, that the phenomenon persists across different $(B,L,n)$ designs and for a grid approximation to a continuous log-significand Cram\'er--von Mises discrepancy, and that corrected procedures retain power against a genuine marginal departure in the wrapped-Gaussian design. Finally, we ask what information remains after the marginal target has been matched: temporal order and conditioning variables can distinguish mechanisms, whereas any difference erased by the map from magnitudes to significands is unavailable to digit-only inference.

The controlled constructions are deliberate: they isolate the source of miscalibration without the confounding features of a particular audit or election. They do not imply that Benford screening is generally invalid. Digit evidence can be useful when a substantive legitimate-process model predicts Benford behavior, a relevant alternative predicts a departure, and the sampling calibration reflects how observations were obtained \citep{barabesi2018,barabesi2026}. The narrower claim is that those three ingredients should not be conflated.

The remainder of the paper develops this argument and ends with a short set of practical recommendations. Technical calculations and secondary numerical results are collected in the separate Supplementary Material.

\section{What a Benford test observes}
\label{sec:observes}

\subsection{Significand, log-significand, and first digit}

For a base $b\ge2$, every $X>0$ can be written uniquely as
\begin{equation}
X=b^{K+U},
\qquad
K=\lfloor\log_bX\rfloor\in\mathbb Z,
\qquad
U=\{\log_bX\}\in[0,1).
\label{eq:KU}
\end{equation}
Here $\{x\}=x-\lfloor x\rfloor$ denotes fractional part. The base-$b$ significand is
\[
S_b(X)=b^U\in[1,b),
\]
while $U$ is its \emph{log-significand}. The full Benford law is equivalent to
\begin{equation}
U\sim\Unif(0,1).
\label{eq:fullbenford}
\end{equation}
From this point onward we set $b=10$. The first-digit event $D_1=d$ is the interval event
\[
U\in A_d,
\qquad
A_d=[\log_{10}d,\log_{10}(d+1)),
\]
so that
\[
p_d
=
\Pp(D_1=d)
=
\Pp(U\in A_d)
=
|A_d|.
\]

We use three related terms. A process is \emph{pointwise marginally Benford} when $\Law(U_{s,t})=\Unif(0,1)$ for every recorded pair $(s,t)$. A collection is \emph{pooled Benford} when the distribution obtained by selecting an observation according to the stated pooling weights is uniform, even if the component distributions are not. A regime is \emph{conditionally Benford} when the uniform law holds after conditioning on that regime. Pointwise marginal or pooled Benfordness is weaker than independence and says nothing about the joint law of the recorded sequence.

\subsection{Pearson's statistic and the sampling covariance}

Let
\[
Y_j
=
\bigl(
\mathbf 1\{U_j\in A_1\},\ldots,
\mathbf 1\{U_j\in A_9\}
\bigr)^\top,
\qquad
\widehat p=\frac1n\sum_{j=1}^nY_j,
\]
and let $p=(p_1,\ldots,p_9)^\top$ and $D_p=\operatorname{diag}(p_1,\ldots,p_9)$. We write the conventional Pearson discrepancy as
\begin{equation}
X_n^2
=
n(\widehat p-p)^\top D_p^{-1}(\widehat p-p)
=
\sum_{d=1}^9\frac{(N_d-np_d)^2}{np_d}.
\label{eq:pearson}
\end{equation}
The notation separates the statistic $X_n^2$ from the distribution used to calibrate it.

Under independent Benford sampling,
\[
\sqrt n(\widehat p-p)
\Rightarrow
N_9(0,\Omega_{\mathrm{iid}}),
\qquad
\Omega_{\mathrm{iid}}=D_p-pp^\top.
\]
Consequently,
\[
X_n^2\Rightarrow\chi^2_8.
\]
The reference law is therefore a statement about the covariance of the empirical histogram, not merely about its mean.

More generally, suppose an observation design yields
\begin{equation}
\sqrt n(\widehat p-p)
\Rightarrow
N_9(0,\Omega),
\qquad
\Omega\mathbf 1=0.
\label{eq:clt-general}
\end{equation}
The relevant parameter space is the multinomial contrast space
\[
\mathcal C=\{x\in\mathbb R^9:\mathbf 1^\top x=0\}.
\]
Indeed, $D_p^{-1/2}\Omega D_p^{-1/2}$ has the null vector
$D_p^{1/2}\mathbf 1=(\sqrt{p_1},\ldots,\sqrt{p_9})^\top$. Let
$\lambda_1,\ldots,\lambda_r$, with $r\leq8$, denote its positive eigenvalues on the corresponding contrast space. Then the continuous mapping theorem gives
\begin{equation}
X_n^2
\Rightarrow
\sum_{j=1}^r\lambda_j\chi^2_{1,j},
\label{eq:weightedchisq}
\end{equation}
where the $\chi^2_{1,j}$ are independent. Under i.i.d. multinomial sampling $r=8$ and all eight positive eigenvalues equal one, recovering $\chi^2_8$. Under serial dependence, clustering, or additional degeneracy they need not.

Equation~\eqref{eq:weightedchisq} is classical in spirit rather than a new theorem of this paper. Closely related weighted-chi-square limits and design-effect corrections appear in the survey-sampling literature \citep{raoscott1981,raoscott1984,raoscott1987}, while \citet{tavarealtham1983} derives corrections for Markov-dependent categorical observations. Its role here is to make the statistical object explicit before specializing to Benford's law.

The same decomposition also clarifies the latent-composition example. Without assuming that $\E(\widehat p)=p$,
\begin{equation}
\E(X_n^2)
=
n\,\operatorname{tr}
\!\left[
D_p^{-1}\operatorname{Var}(\widehat p)
\right]
+
n\{\E(\widehat p)-p\}^\top
D_p^{-1}
\{\E(\widehat p)-p\}.
\label{eq:pearson-mean-general}
\end{equation}
Here $\operatorname{Var}(\widehat p)$ is the covariance under the actual observation design. The first term is a sampling-variance contribution; the second is displacement from the target. The primary benchmark fixes the target exactly and varies the first term. The random-composition experiment later shows how both terms can matter.

\subsection{From a marginal statement to a substantive claim}

A mechanism $M$ and observation design $\mathcal D_n$ generate
\[
X_{1:n}\sim P_{M,\mathcal D_n}.
\]
The Benford target concerns $\Law_M(U_j)$. The sampling law of a statistic concerns
\[
\Law_{M,\mathcal D_n}\{T(U_{1:n})\}.
\]
These are different levels of description:
\begin{equation}
\Law(U_j)=\Unif(0,1)
\quad\not\Rightarrow\quad
\Law(U_{1:n})=\Unif(0,1)^{\otimes n}.
\label{eq:marginal-not-joint}
\end{equation}

A forensic interpretation introduces another level. Let $H_F$ denote a substantive integrity hypothesis whose implications for the observed data must be specified. Schematically,
\[
M
\longrightarrow X_{1:n}
\longrightarrow U_{1:n}
\longrightarrow T
\longrightarrow \text{substantive conclusion}.
\]
Benford's law constrains the middle representation; it does not supply the last arrow. A digit discrepancy can be evidential when legitimate and manipulated models predict different digit behavior, as in established and recent antifraud procedures \citep{barabesi2018,barabesi2026}. Conversely, digit conformity cannot by itself certify integrity.

Many mechanisms can generate Benford behavior, including random products, deterministic dynamics, mixtures, fragmentation, and invariant Markov systems \citep{hill1995,millernigrini2008,bergerhill2011,becker2018,burgossantos2021}. Morrison's multiplication game provides a strikingly different example in which Benford/Haar randomization is strategically invariant \citep{morrison2010}; because this example is not needed for the calibration argument, its numerical illustrations are provided in the Supplementary Material, Section~S3.

\section{Same Benford target, different sampling laws}
\label{sec:benchmark}

\subsection{Four constructions with the same Benford target}

The benchmark uses $B=400$ independent sequences of length $L=250$, so the total record count is $n=BL=100{,}000$ per realization. When convenient, the flattened index $j=1,\ldots,n$ corresponds to $j=(s-1)L+t$. All transitions are defined within sequence boundaries. This large $n$ makes sparse-cell and small-sample $\chi^2$ approximations irrelevant in the i.i.d. reference, so departures from the conventional calibration can be attributed to sampling structure rather than to small expected counts. The final primary comparison uses 2,000 independent Monte Carlo realizations; the persistence sensitivity analysis uses 1,000 realizations at each value of $\rho$, and the random-composition experiment uses 2,000 realizations. No generator is calibrated to the Pearson statistic: $\rho=0.5$ is a fixed illustrative persistence level, and the irrational rotation uses the prespecified golden-ratio conjugate.

The first construction is an i.i.d. reference,
\begin{equation}
U_{s,t}\overset{\mathrm{iid}}{\sim}\Unif(0,1).
\label{eq:iid}
\end{equation}

The second is a stochastic multiplicative process represented on the logarithmic circle,
\begin{equation}
U_{s,t+1}=(U_{s,t}+\varepsilon_{s,t+1})\bmod1,
\label{eq:markov}
\end{equation}
with each sequence initialized from $\Unif(0,1)$. The innovations are obtained by wrapping centered Gaussian increments with variance $\sigma^2$. Their first circular Fourier coefficient is real and positive,
\begin{equation}
\rho
:=
\E\{e^{2\pi i\varepsilon}\}
=
e^{-2\pi^2\sigma^2}\in(0,1],
\label{eq:rho}
\end{equation}
because the unwrapped increment distribution is centered and symmetric. Thus $\rho$ is a first circular Fourier coefficient---a concentration parameter, not a linear correlation. More generally,
$\psi_k:=\E(e^{2\pi ik\varepsilon})=\rho^{k^2}$.
The primary comparison uses $\rho=0.5$. Uniform measure is invariant under addition on the circle, so $U_{s,t}$ is exactly uniform for every $s,t$, independently of $\rho$.

The third construction is a deterministic irrational rotation with randomized phase,
\begin{equation}
U_{s,t+1}=(U_{s,t}+\theta)\bmod1,
\qquad
\theta=\frac{\sqrt5-1}{2},
\label{eq:rotation}
\end{equation}
where the initial phase $U_{s,1}$ is drawn independently from $\Unif(0,1)$ across sequences. The probability statements in the benchmark are therefore with respect to this phase randomization. Conditional on $U_{s,1}$ the subsequent path is deterministic. Every one-time marginal is uniform, while along any fixed phase an irrational rotation is asymptotically equidistributed as $L\to\infty$. We deliberately use the golden-ratio conjugate, a canonical badly approximable irrational with favorable low-discrepancy properties \citep{kuipersniederreiter1974}. The extreme finite-$L$ under-dispersion below is therefore an illustrative low-discrepancy construction, not a generic claim about all irrational rotations; finite-$L$ discrepancy depends on $\theta$, the digit intervals, and the phase.

The fourth construction introduces latent heterogeneity. Exactly half of the sequences have $Z_s=0$ and half have $Z_s=1$, with observations conditionally independent and
\begin{equation}
f_{U\mid Z=0}(u)=2u,
\qquad
f_{U\mid Z=1}(u)=2(1-u),
\qquad 0<u<1.
\label{eq:latent}
\end{equation}
Neither conditional regime is Benford. Their equal-weight mixture is exactly uniform because
\[
\frac12(2u)+\frac12\,2(1-u)=1.
\]
The balanced design is an idealized control: it holds the realized mixture weight fixed at one half so that the primary comparison isolates conditional heterogeneity rather than random composition noise. Section~\ref{sec:composition} restores random regime composition and treats it as part of the sampling design.

\begin{proposition}[Exact marginal agreement]
\label{prop:marginal}
The i.i.d., stochastic multiplicative, and randomized irrational-rotation constructions are marginally Benford: $U_{s,t}\sim\Unif(0,1)$ at every observation time. The balanced latent construction is pooled Benford under the stated equal weights, although neither conditional regime is Benford.
\end{proposition}

\begin{proof}
The i.i.d. statement is immediate. For the multiplicative process and the rotation, uniform measure on the circle is invariant under translation, so a uniform initial state remains uniform after every update. For the latent construction, the pooled density is the equal-weight mixture in equation~\eqref{eq:latent}, which is identically one.
\end{proof}

The proposition fixes the one-observation target while leaving the sampling law of empirical summaries unrestricted. Figure~\ref{fig:marginal-sampling} shows the central contrast. The first-digit frequencies are visually indistinguishable from the same population Benford probabilities, while the sampling distributions of the Pearson statistic differ sharply.

\begin{figure}[htbp]
\centering
\begin{subfigure}[t]{0.49\textwidth}
\centering
\includegraphics[width=\textwidth]{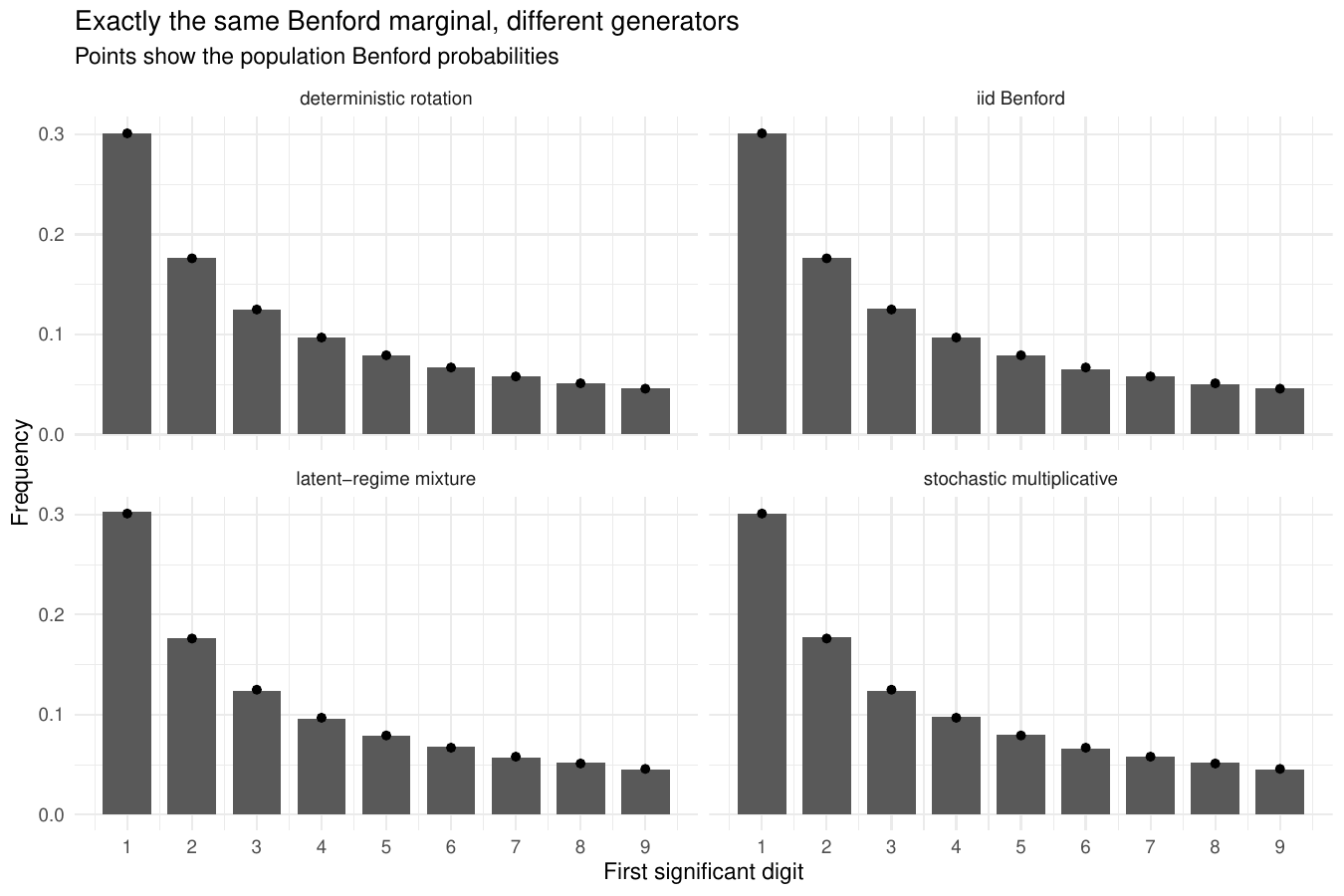}
\caption{Representative first-digit frequencies.}
\end{subfigure}\hfill
\begin{subfigure}[t]{0.49\textwidth}
\centering
\includegraphics[width=\textwidth]{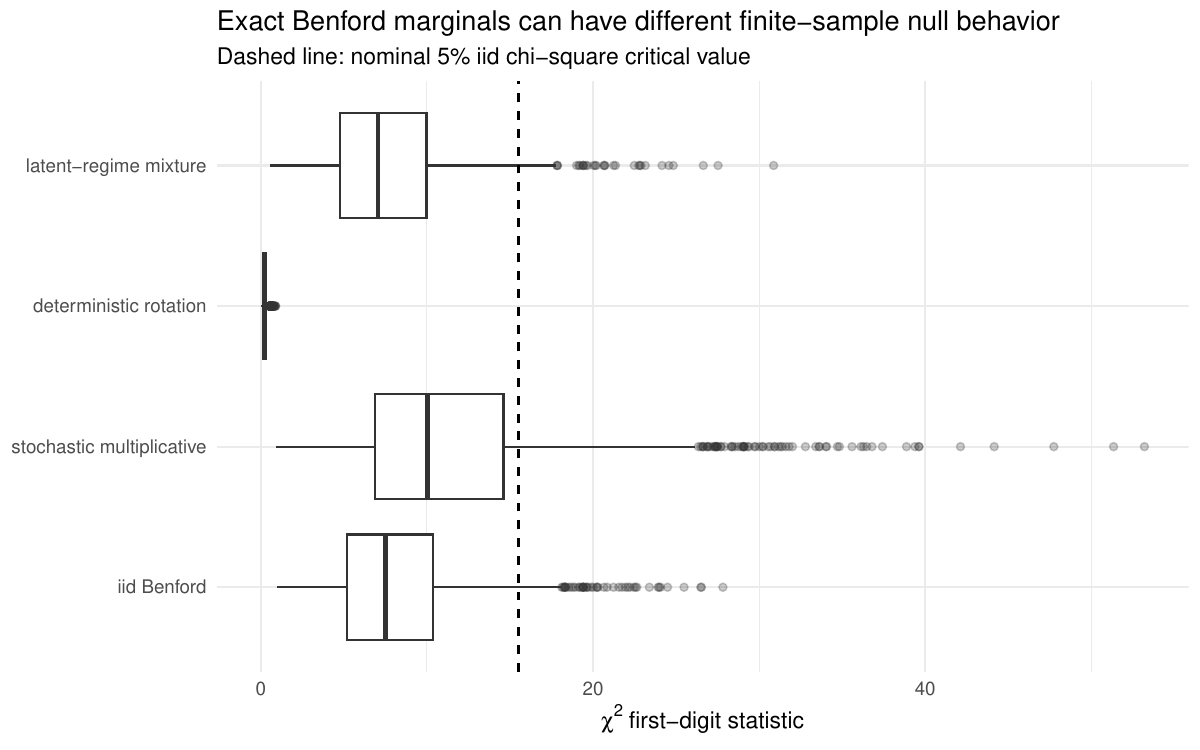}
\caption{Sampling distribution of $X_n^2$.}
\end{subfigure}
\caption{The four constructions agree exactly on the relevant one-observation or pooled Benford target but have different finite-sample behavior. In the embedded generator labels, ``deterministic rotation'' denotes the rotation with independently randomized initial phase and ``latent-regime mixture'' denotes the exactly balanced mixture. The dashed line in panel (b) is the conventional 5\% $\chi^2_8$ critical value.}
\label{fig:marginal-sampling}
\end{figure}

\subsection{A nominal i.i.d. calibration can reject a marginally Benford process}

Table~\ref{tab:main} reports both the analytical expectation of $X_n^2$ and the Monte Carlo results. The analytical calculations use the known covariance structure of the digit indicators; details are in the Supplementary Material, Section~S1. The agreement with simulation is close in every construction.

\begin{table}[htbp]
\centering
\caption{Common-target Benford benchmark. Monte Carlo entries are based on 2,000 realizations of $n=100{,}000$ observations. Standard deviations are in parentheses.}
\label{tab:main}
\begin{tabular}{lrrr}
\toprule
Construction & $\E(X_n^2)$ theory & MC $X_n^2$ & Reject .05 \\
\midrule
i.i.d. Benford & 8.000 & 8.111 (4.086) & 0.052 \\
Stochastic multiplicative & 11.711 & 11.511 (6.620) & 0.216 \\
Randomized irrational rotation & 0.219 & 0.223 (0.132) & 0.000 \\
Balanced latent mixture & 7.679 & 7.676 (3.941) & 0.038 \\
\bottomrule
\end{tabular}
\end{table}

With 2,000 primary replications, the Monte Carlo standard error of each rejection rate in Table~\ref{tab:main} is below 0.01. The analytical column concerns the expectation of the Pearson statistic; the rejection probabilities use the simulated sampling distributions.

The i.i.d. construction is calibrated as expected. The wrapped-Gaussian multiplicative construction has the same one-time Benford probabilities, but at $\rho=0.5$ the nominal test rejects 21.6\% of samples. The randomized rotation produces the opposite effect: its low-discrepancy orbit is strongly under-dispersed relative to independent sampling, so none of the 2,000 realizations crosses the conventional critical value. Neither result is a statement about the Benford marginal, which is fixed by construction. They are statements about the sampling variability around that marginal.

The persistence sweep reinforces the point without suggesting that all dependence inflates the statistic. In the wrapped-Gaussian construction the naive rejection rate rises above 50\% by $\rho=0.75$ and approaches 90\% at $\rho=0.9$, whereas the rotation shows that a different joint structure can produce severe deflation. Marginal conformity alone does not determine a null distribution.

\subsection{Design-aware calibration}
\label{sec:design-aware}

The weighted-chi-square representation in equation~\eqref{eq:weightedchisq} also suggests a correction when the independent sampling units are observed. In the benchmark those units are the $B$ sequences, not the $BL$ individual records. For sequence $s$, define
\[
G_s=\sum_{t=1}^L(Y_{s,t}-p).
\]
The following elementary sequence-level limit states the asymptotic regime used by the correction.

\begin{proposition}[Sequence-level Gaussian limit]
\label{prop:cluster-clt}
Suppose $L$ is fixed, $G_1,\ldots,G_B$ are i.i.d. across sequences, $\E(G_s)=0$, and $\E\|G_s\|^2<\infty$. As $B\to\infty$,
\[
\sqrt{BL}(\widehat p-p)
=
\frac{1}{\sqrt L}\frac{1}{\sqrt B}\sum_{s=1}^B G_s
\Rightarrow
N_9(0,\Omega_L),
\qquad
\Omega_L=\frac{1}{L}\operatorname{Var}(G_s).
\]
If $D_p^{-1/2}\Omega_LD_p^{-1/2}$ has rank $r$ on the multinomial contrast space, then
\[
X_n^2\Rightarrow\sum_{j=1}^r\lambda_j\chi^2_{1,j},
\]
where $\lambda_1,\ldots,\lambda_r$ are its positive eigenvalues.
\end{proposition}

\begin{proof}
The first statement is the multivariate central limit theorem applied to the independent sequence summaries $G_s$. The quadratic-form limit follows from the continuous mapping theorem exactly as in equation~\eqref{eq:weightedchisq}.
\end{proof}

For identically distributed independent sequences, the covariance of the empirical digit histogram can therefore be estimated from the between-sequence covariance of the $G_s$. The proposition is an asymptotic statement in the number of independent sequences $B$, not in the number of within-sequence records alone. It also covers the randomized rotation for fixed $L$: conditional on its phase a sequence is deterministic, but the phase-randomized summaries $G_s$ are i.i.d. and bounded across sequences. If $L$ grows with $B$, additional conditions on the triangular array of sequence summaries are needed and are not claimed here.

In the balanced latent design, where the 50/50 regime allocation is fixed rather than random, we estimate the covariance within each regime and combine the two components with their fixed weights.

Let $\widehat\lambda_1,\ldots,\widehat\lambda_{\widehat r}$ be the positive estimated design-effect eigenvalues corresponding to equation~\eqref{eq:weightedchisq}. A first-order Rao--Scott-type correction replaces the heterogeneous spectrum by its mean,
\[
\overline\lambda=\frac{1}{\widehat r}\sum_{j=1}^{\widehat r}\widehat\lambda_j,
\qquad
\frac{X_n^2}{\overline\lambda}\approx\chi^2_{\widehat r}.
\]
We also use a second-order moment-matching approximation, in the spirit of Rao--Scott and Satterthwaite \citep{raoscott1984,satterthwaite1946}, that retains the first two moments of the estimated weighted sum:
\begin{equation}
\widehat a=
\frac{\sum_j\widehat\lambda_j^2}{\sum_j\widehat\lambda_j},
\qquad
\widehat\nu=
\frac{(\sum_j\widehat\lambda_j)^2}{\sum_j\widehat\lambda_j^2},
\qquad
\frac{X_n^2}{\widehat a}\approx\chi^2_{\widehat\nu}.
\label{eq:rs2}
\end{equation}
The approximation treats the estimated spectrum as a plug-in quantity; it is not an exact finite-sample pivot. For comparison, we also form a sequence-aware Wald statistic on eight independent digit contrasts using the same estimated covariance. These procedures are not proposed as universal Benford tests; they are design-aware calibrations for the specific setting in which independent sequence boundaries are known. In the implementation, the nine-category problem is represented by eight free cell proportions, the estimated covariance is symmetrized, numerically negligible negative eigenvalues are truncated at zero, and the Wald statistic uses a Moore--Penrose inverse with a fixed numerical tolerance. These choices are documented in the companion code.

Table~\ref{tab:design-aware} reports the finite-sample result. The second-order correction is close to the nominal 5\% level in all four reported constructions: 5.05\% for i.i.d. sampling, 4.75\% for the multiplicative process, 5.10\% for the randomized rotation, and 4.55\% for the balanced latent mixture. The sequence-aware Wald test is also close, although mildly anti-conservative in these finite-$B$ experiments (5.45--6.40\%). By contrast, the first-order correction is less reliable when the eigenvalue spectrum is far from a common scalar multiple of the i.i.d. covariance.

\begin{table}[htbp]
\centering
\caption{Monte Carlo rejection rates at nominal 5\% in an independently seeded calibration experiment under the conventional i.i.d. Pearson reference and three sequence-aware calibrations. Small differences from Table~\ref{tab:main} are Monte Carlo variation. The final columns summarize the estimated design-effect spectrum.}
\label{tab:design-aware}
\begin{tabular}{lrrrrrr}
\toprule
Construction & Naive & RS1 & RS2 & Wald & Mean $\widehat\lambda$ & RS2 df \\
\midrule
i.i.d. Benford & 0.052 & 0.052 & 0.051 & 0.062 & 1.000 & 7.83 \\
Stochastic multiplicative & 0.217 & 0.066 & 0.048 & 0.056 & 1.463 & 6.22 \\
Randomized irrational rotation & 0.000 & 0.080 & 0.051 & 0.055 & 0.027 & 5.08 \\
Balanced latent mixture & 0.039 & 0.050 & 0.046 & 0.064 & 0.961 & 7.71 \\
\bottomrule
\end{tabular}
\end{table}

The spectrum is itself informative. For the multiplicative construction the average design effect is 1.46 and the largest estimated eigenvalue is about 2.99; for the rotation the corresponding mean is only 0.027. The second-order effective degrees of freedom therefore fall below eight because the distortion is not merely a common variance multiplier. This is why a scalar first-order adjustment does not fully repair every design.

Figure~\ref{fig:rho-calibration} makes the practical implication visible. As circular persistence increases, the conventional Pearson rejection rate rises sharply, whereas the second-order correction and the sequence-aware Wald calibration remain close to the nominal level across the reported simulation range.

\begin{figure}[htbp]
\centering
\includegraphics[width=0.76\textwidth]{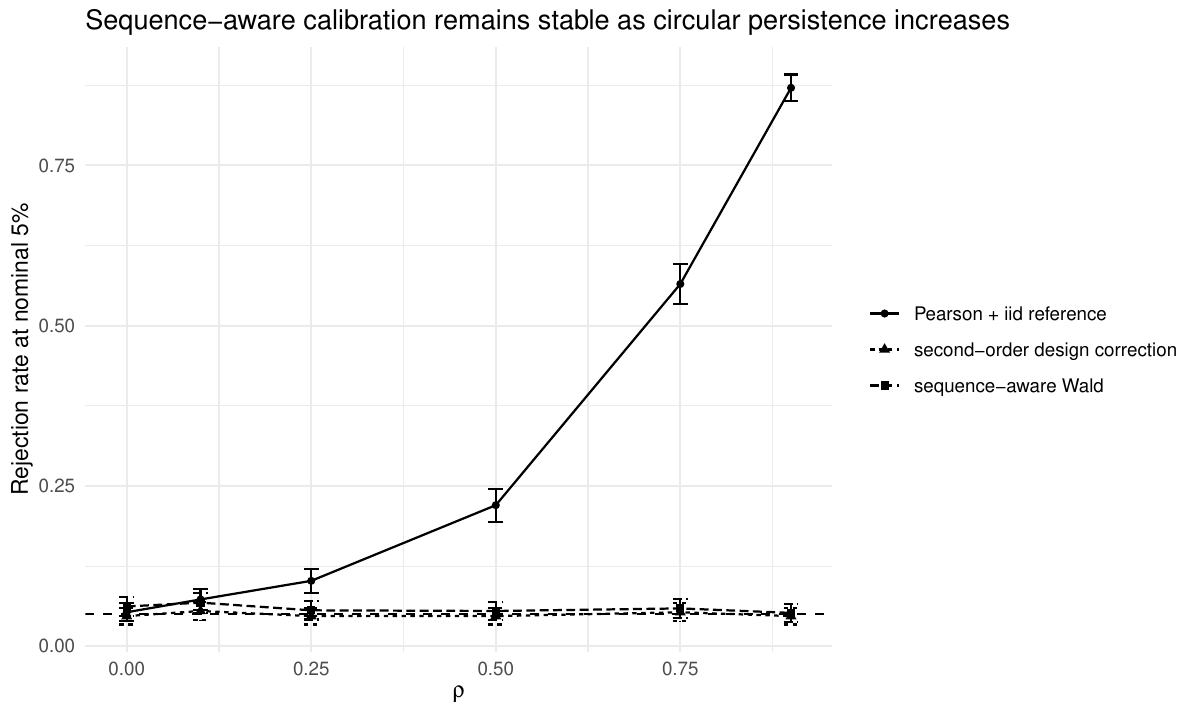}
\caption{Rejection rates at nominal 5\% as the first circular Fourier coefficient $\rho$ increases in the wrapped-Gaussian Markov construction. The marginal log-significand remains exactly uniform for every $\rho$. The second-order design correction and sequence-aware Wald statistic use the observed sequence structure.}
\label{fig:rho-calibration}
\end{figure}

\subsection{Robustness to sequence length, independent-sequence count, and statistic}
\label{sec:robustness}

The primary benchmark uses $B=400$, $L=250$, and hence $n=100{,}000$. To check that the calibration contrast is not a large-$n$ artifact, we rerun the wrapped-Gaussian circular Markov construction at $\rho=0.5$ under two design grids. First, $B=400$ is fixed while $L$ varies from 10 to 250, so $n$ ranges from 4,000 to 100,000. Second, $n=100{,}000$ is held fixed while $B$ ranges from 100 to 2,000 and $L$ changes inversely. Each design uses 500 Monte Carlo replications.

The naive rejection rate does not increase monotonically with total sample size. With $B=400$, it is already 19.6\% at $n=4{,}000$ and ranges from 19.6\% to 26.8\% over the 25-fold change in $n$. The second-order correction ranges from 3.8\% to 6.4\%. Holding $n=100{,}000$ fixed while changing the number of independent sequences by a factor of 20 gives naive rejection rates between 20.4\% and 23.4\%, compared with 4.2--4.8\% for the second-order correction. These results point to the covariance induced by the within-sequence design rather than total record count as the relevant feature in this construction.

\begin{table}[htbp]
\centering
\caption{Sensitivity of rejection rates to $B$, $L$, and $n$ in the wrapped-Gaussian circular Markov construction at $\rho=0.5$. Entries summarize the range over each prespecified design grid; the full grid is reported in the Supplementary Material, Section~S4.}
\label{tab:design-size-summary}
\begin{tabular}{lrrr}
\toprule
Design grid & Pearson i.i.d. & RS2 & Sequence Wald \\
\midrule
$B=400$, $n=4{,}000$--$100{,}000$ & 0.196--0.268 & 0.038--0.064 & 0.056--0.072 \\
$n=100{,}000$, $B=100$--$2{,}000$ & 0.204--0.234 & 0.042--0.048 & 0.040--0.094 \\
\bottomrule
\end{tabular}
\end{table}

The Wald statistic is somewhat more sensitive to the number of independent sequences: its rejection rate reaches 9.4\% in the $B=100$, $L=1{,}000$ design. This finite-cluster behavior is an informative limitation rather than a contradiction. Estimating and inverting an eight-dimensional covariance matrix from relatively few independent clusters is itself an inferential problem; a design-aware method does not remove the need for enough independent sampling units. We do not claim a universal minimum $B$: the relevant threshold depends on the dimension, heterogeneity of the design spectrum, and stability of the estimated covariance.

The same issue is not specific to first-digit binning. The ideal continuous log-significand Cram\'er--von Mises discrepancy is
\[
W_n^2=n\int_0^1\{F_n(u)-u\}^2\,du.
\]
In the numerical experiment we use its explicit grid approximation
\begin{equation}
\widetilde W_{n,M}^2
=
\frac{n}{M+1}\sum_{m=1}^M\{F_n(u_m)-u_m\}^2,
\qquad
u_m=\frac{m}{M+1},
\qquad M=49.
\label{eq:cvm}
\end{equation}
Under the conventional i.i.d. reference, the multiplicative construction rejects 22.2\% of samples and the randomized irrational rotation never rejects. Estimating the sequence-level covariance of the empirical CDF grid and applying the same second-order moment-matching strategy gives rejection rates of 4.8\% and 4.4\%, respectively. The i.i.d. and balanced-mixture corrected rates are 6.2\% and 6.4\%. With 500 replications these deviations from 5\% are of the same order as Monte Carlo uncertainty. The exercise is deliberately a robustness check, not a proposal for an optimal dependent-data CvM procedure: changing the discrepancy statistic does not, by itself, remove the sampling-design problem.

\section{A genuine marginal departure: size and power}
\label{sec:power}

A correction is useful only if it improves null calibration without removing sensitivity to a genuine marginal departure. We therefore keep the same latent $\rho=0.5$ circular Markov construction and alter only its one-time marginal. If $V$ denotes the original uniform log-significand, we transform it through the inverse CDF of
\begin{equation}
f_\epsilon(u)=1+\epsilon(2u-1),\qquad
F_\epsilon(u)=(1-\epsilon)u+\epsilon u^2,
\label{eq:tilt-alt}
\end{equation}
so that $\epsilon=0$ is the exactly Benford dependent null and $\epsilon>0$ is a smooth marginal departure. We use 500 replications at each prespecified value $\epsilon\in\{0,0.005,0.01,0.02,0.03,0.05\}$.

Figure~\ref{fig:power} makes the distinction between size distortion and power explicit. At $\epsilon=0$, the conventional i.i.d. Pearson and CvM references reject about one quarter of samples, whereas the design-aware procedures are near their intended level. In this wrapped-Gaussian design, their rejection probabilities then increase monotonically over the prespecified $\epsilon$ grid: at $\epsilon=.02$, for example, Pearson RS2 rejects 43.8\%, the sequence-aware Wald test 41.2\%, and design-corrected CvM 62.8\%; by $\epsilon=.03$ the corresponding rates are 83.0\%, 79.6\%, and 93.4\%. Thus the higher naive curves cannot be interpreted as superior power, because they begin from a badly inflated null size. Valid calibration does not require discarding the marginal signal.

\begin{figure}[htbp]
\centering
\includegraphics[width=0.86\textwidth]{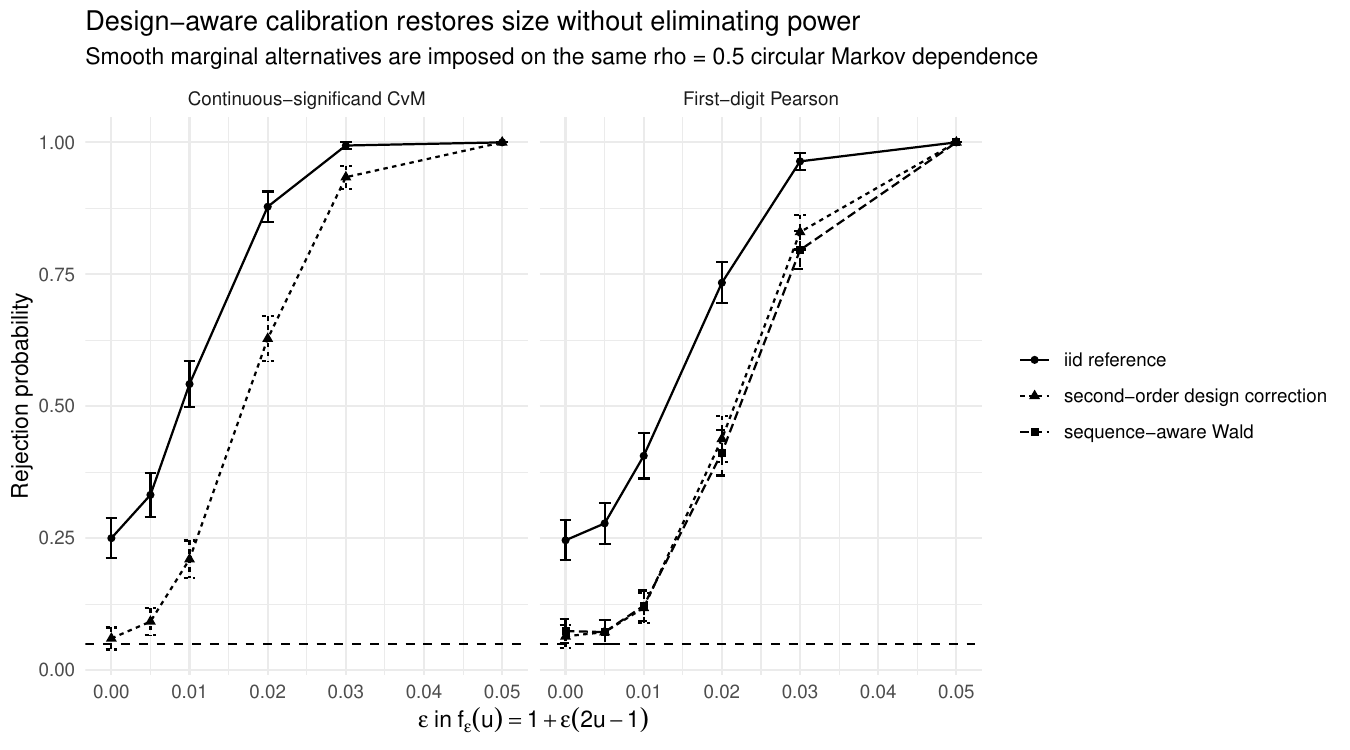}
\caption{Size and power under the smooth marginal alternative in equation~\eqref{eq:tilt-alt}. Both panels use the same latent $\rho=0.5$ circular Markov construction. The horizontal dashed line marks 5\%. Error bars are Monte Carlo 95\% intervals based on 500 replications.}
\label{fig:power}
\end{figure}

\section{What additional information distinguishes the mechanisms}
\label{sec:recover}

If several mechanisms share a marginal law, discrimination requires observations on which they make different predictions. In the benchmark, temporal order distinguishes the dynamic constructions, while conditioning distinguishes latent aggregation.

\subsection{Order information}

For each sequence we summarize changes in logarithmic significand through the first circular transition moment. With equal sequence length $L$, the total number of within-sequence transitions is $N_{\mathrm{tr}}=B(L-1)$, and
\begin{equation}
\widehat m
=
\frac{1}{N_{\mathrm{tr}}}
\sum_{s=1}^B\sum_{t=1}^{L-1}
\exp\{2\pi i(U_{s,t+1}-U_{s,t})\}.
\label{eq:mhat}
\end{equation}
A raw moment can be nonzero simply because a sequence has a non-uniform composition. We therefore subtract its exact conditional expectation under random permutation within each sequence and report
\begin{equation}
C_{\mathrm{centered}}
=
\left|
\widehat m-\E_{\mathrm{shuffle}}(\widehat m)
\right|.
\label{eq:Ccentered}
\end{equation}
The closed-form permutation expectation is given in the Supplementary Material, Section~S4. The statistic is used only as a mechanism-specific diagnostic: it asks whether the observed ordering contains circular transition structure beyond that induced by the within-sequence composition.

Figure~\ref{fig:extra-info}a shows that $C_{\mathrm{centered}}$ is near zero for the i.i.d. and balanced latent-mixture constructions, about 0.49 for the multiplicative process at $\rho=0.5$, and essentially one for the randomized irrational rotation. Order information therefore separates dynamics that a first-digit histogram cannot.

\subsection{Conditioning and aggregation}

The latent construction calls for a different diagnostic. Figure~\ref{fig:extra-info}b shows the distribution of $U$ pooled and conditional on the latent state. The pooled density is uniform, while the two conditional densities are strongly non-uniform and mirror one another. The mechanism is revealed by conditioning, not by retaining more digits from the pooled sample.

\begin{figure}[htbp]
\centering
\begin{subfigure}[t]{0.49\textwidth}
\centering
\includegraphics[width=\textwidth]{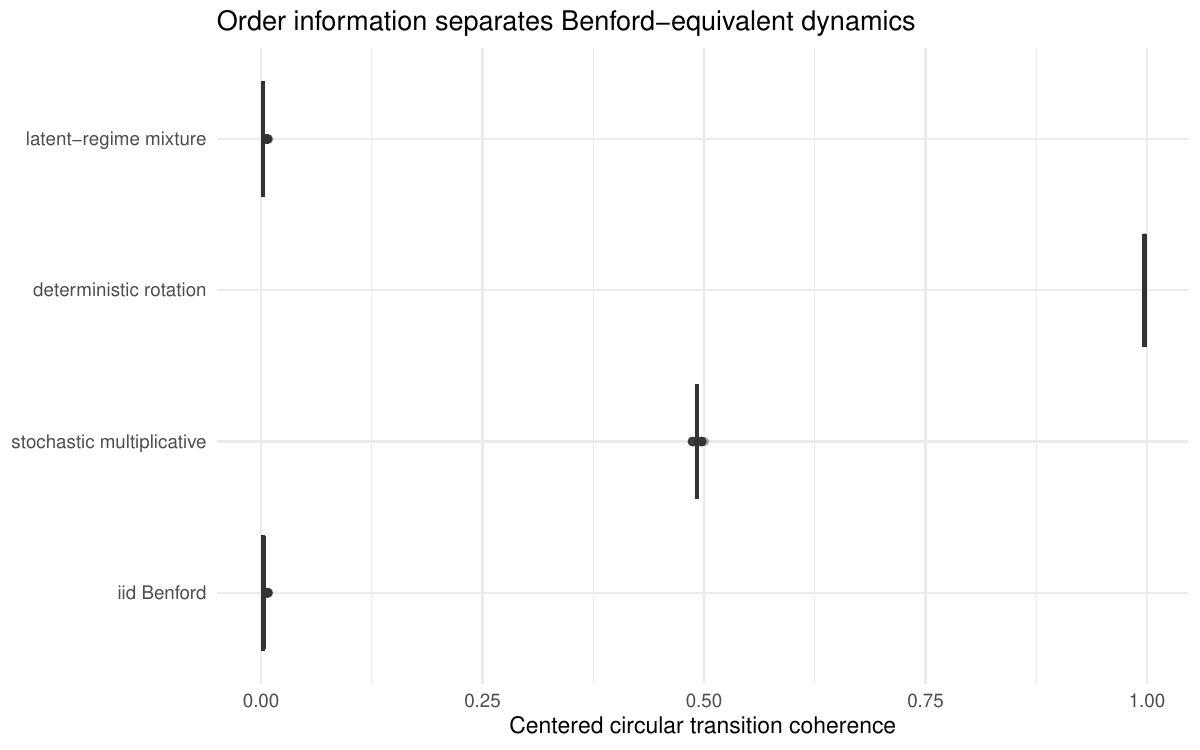}
\caption{Order information.}
\end{subfigure}\hfill
\begin{subfigure}[t]{0.49\textwidth}
\centering
\includegraphics[width=\textwidth]{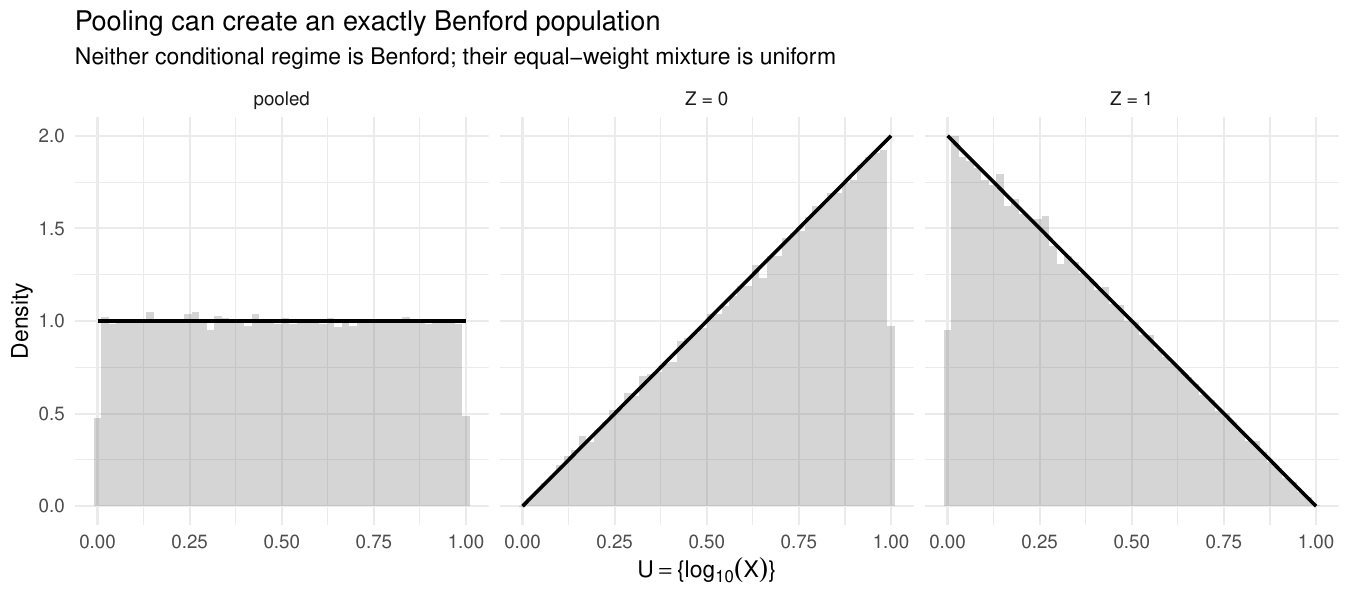}
\caption{Conditioning on the latent regime.}
\end{subfigure}
\caption{Additional information discriminates mechanisms that share the same Benford first-digit target. Panel (a) uses temporal order; panel (b) uses the conditioning variable that generated the mixture. The shortened embedded labels use ``deterministic rotation'' for the random-phase rotation and ``latent-regime mixture'' for the balanced mixture.}
\label{fig:extra-info}
\end{figure}

\subsection{Random cluster composition}
\label{sec:composition}

Balanced pooling was imposed in the primary benchmark so that the two latent regimes contribute equal numbers of observations in every realization. In applications, the number of observed clusters from each regime will usually fluctuate. Restoring that randomness creates a second gap: a superpopulation mixture can be exactly Benford while a realized clustered sample is not exactly balanced.

Let $\widehat\pi$ be the realized fraction of sequences from the first latent regime. This example separates two null statements that coincide under i.i.d. sampling but not under clustered composition:
\[
H_0^{\mathrm{marginal}}:\quad \Pp(U_{s,t}\in A_d)=p_d,
\qquad
H_0^{\mathrm{iid}}:\quad (N_1,\ldots,N_9)\sim\operatorname{Multinomial}(n,p).
\]
The superpopulation construction satisfies the first statement, whereas the second ignores the random sequence-level composition. Conditional on $\widehat\pi$, the pooled first-digit probabilities depart from the equal-weight target whenever $\widehat\pi\ne1/2$. The exact conditional expectation of the Pearson statistic is derived in the Supplementary Material, Section~S2. For the two triangular regimes used here, its dominant displacement term is
\begin{equation}
1.28424\,n\left(\widehat\pi-\frac12\right)^2.
\label{eq:composition-leading}
\end{equation}
The $n$ multiplier makes ordinary cluster-composition fluctuations visible to a test calibrated as if all $n$ observations were independent draws from one fixed multinomial distribution.

With 400 sequences, random Bernoulli assignment to the two regimes yields a mean Pearson statistic of 89.46 and a 76.0\% rejection rate. For comparison, the primary fixed-balance benchmark has Monte Carlo mean 7.676 and rejection rate 3.8\%, with theoretical expectation 7.6789. Figure~\ref{fig:composition} overlays the random-composition simulations with the exact conditional expectation, not a fitted smoother. Moreover, since $\operatorname{Var}(\widehat\pi)=1/(4B)$ under Bernoulli assignment and $n=BL$, equation~\eqref{eq:composition-leading} predicts an average imbalance contribution of approximately
\[
1.28424\,\frac{L}{4}=80.27
\]
when $L=250$, before adding the within-regime variance term. The corresponding unconditional approximation, 87.94, is close to the simulated mean 89.46. The Monte Carlo standard error of that simulated mean is about $111.1/\sqrt{2000}=2.49$, so the 1.52 difference is well within Monte Carlo uncertainty. Thus the effect is governed by the cluster structure, not by treating $n=100{,}000$ records as 100,000 independent draws from a single multinomial distribution.

The sequence-level covariance correction also separates a target discrepancy from a calibration failure in this experiment. Under random Bernoulli composition, the first-order scalar adjustment still rejects 16.2\% of samples because the estimated design-effect spectrum is extremely uneven: its mean is about 10.99, but the second-order effective degrees of freedom are only 1.18. The second-order correction reduces the rejection rate from 76.0\% to 5.0\%, while the sequence-aware Wald test rejects 6.1\%. In the independently seeded fixed-balance calibration experiment of Table~\ref{tab:design-aware}, the corresponding rates are 4.6\% and 6.4\%. These small cross-experiment differences are ordinary Monte Carlo variation. The large naive rejection rate is therefore not evidence that the superpopulation Benford target has failed; it is what results from calibrating a clustered composition experiment as if it were an i.i.d. multinomial sample.

\begin{figure}[htbp]
\centering
\includegraphics[width=0.72\textwidth]{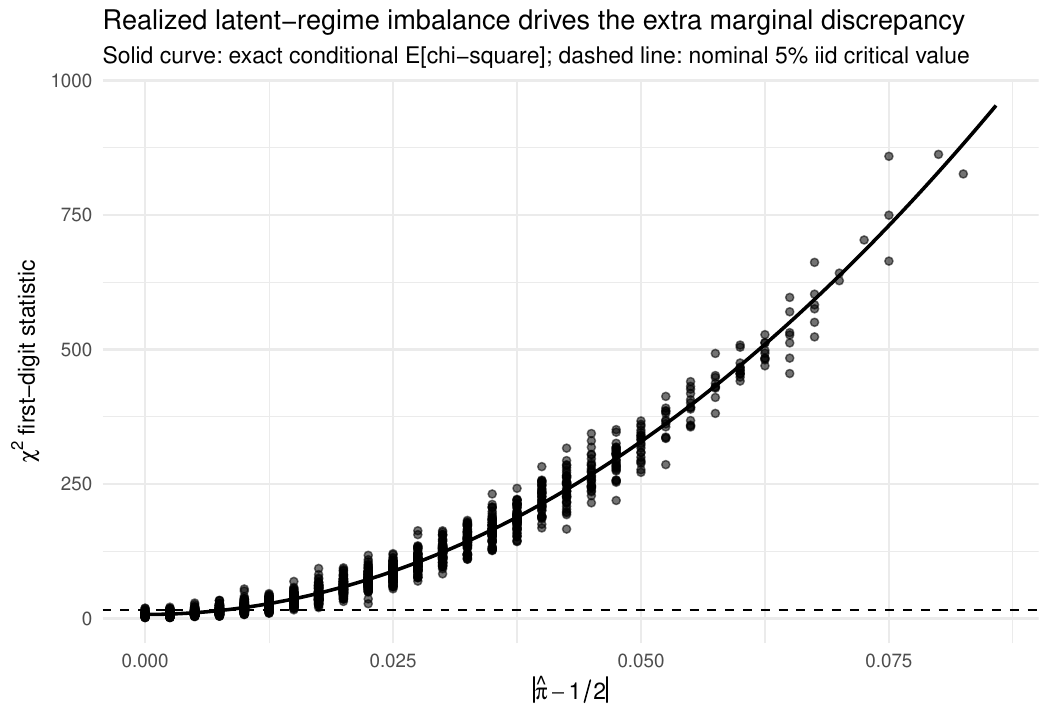}
\caption{First-digit Pearson statistic versus the realized latent-regime imbalance. The solid curve is the exact conditional expectation derived in the Supplementary Material, Section~S2; the dashed line is the conventional 5\% $\chi^2_8$ critical value.}
\label{fig:composition}
\end{figure}

This example is not a claim that arbitrary mixtures should be Benford. It shows that ``the superpopulation mixture is Benford'' and ``the conventional i.i.d. Pearson calibration is valid for the observed clustered sample'' are different statements.

\section{From conformity to integrity: an information boundary}
\label{sec:integrity}

The previous sections show that a correct Benford marginal need not imply a correctly calibrated i.i.d. goodness-of-fit test. For forensic interpretation there is a different limitation: even the complete significand process may contain no information about a relevant difference between two mechanisms.

\paragraph{Observation boundary.}
Let $\mathcal G_n=\sigma(U_1,\ldots,U_n)$ be the information retained by a significand-only analysis. If two mechanisms induce the same law on $\mathcal G_n$---equivalently,
\[
\Law_{M_0}(U_1,\ldots,U_n)
=
\Law_{M_1}(U_1,\ldots,U_n),
\]
---then every $\mathcal G_n$-measurable statistic has the same distribution under the two mechanisms. For a process-level statement, the corresponding requirement is equality of these finite-dimensional laws for every $n$. This is an elementary consequence of measurable transformations, not a new identification theorem. Its role is to mark the information boundary: refining first-digit analysis to second digits or to the complete significand cannot recover differences that lie outside $\mathcal G_n$.

The following construction is best viewed through invariance rather than as a behavioral model of fraud. Write
\[
X_t=10^{K_t+U_t},
\qquad K_t\in\mathbb Z.
\]
Consider the coordinatewise discrete scaling group
\[
(X_1,\ldots,X_n)
\mapsto
(10^{k_1}X_1,\ldots,10^{k_n}X_n),
\qquad k_t\in\mathbb Z.
\]
The vector $U_{1:n}$ is a maximal invariant for this action: two positive magnitude vectors lie on the same orbit exactly when they have the same log-significands coordinate by coordinate. Equivalently, define
\begin{equation}
X_t'=10^{\Delta_t}X_t,
\qquad \Delta_t\in\mathbb Z.
\label{eq:manip}
\end{equation}
Then $X_t'$ and $X_t$ have exactly the same significand observation by observation. Any significand-only test therefore has the same rejection probability along such an orbit; against alternatives that change only the integer log-magnitude component $K_t$, its rejection probability cannot exceed its null rejection probability by using significand information alone. In the numerical illustration, 10\% of observations are shifted upward by two orders of magnitude. The first-digit Pearson statistic and the grid-based CvM discrepancy on $U$ are identical before and after the transformation up to numerical precision, while the distribution of magnitudes changes visibly (Figure~\ref{fig:manip}).

\begin{figure}[htbp]
\centering
\includegraphics[width=0.65\textwidth]{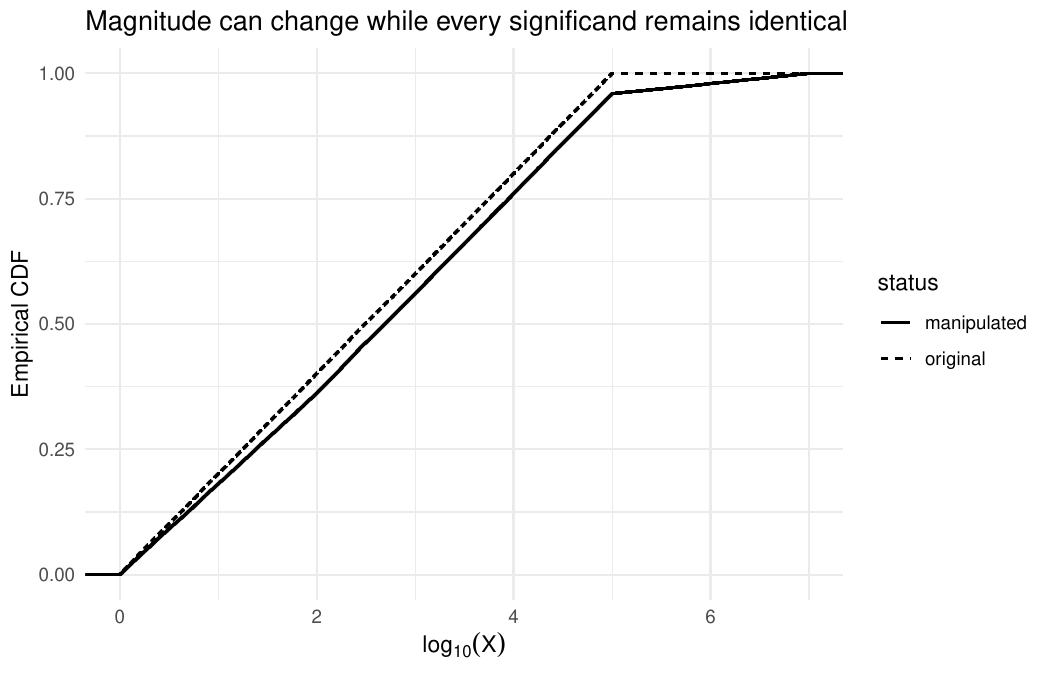}
\caption{A significand-preserving transformation changes magnitudes while leaving the complete significand sequence unchanged. The construction is an impossibility example, not a behavioral model of fraud.}
\label{fig:manip}
\end{figure}

The example should not be read as a claim about how manipulation usually occurs, and the paper does not suggest that real manipulations generally preserve significands. Its purpose is narrower: to show that integrity is not identified from the log-significand information set alone. To make the boundary explicit in the other direction, we add a simple positive control. Starting from an i.i.d. Benford sample, 2\% of observations are replaced by draws from $\Unif(0,0.015)$, a narrow interval near $U=0$ that mimics round-number heaping near powers of ten. This perturbation changes the observed log-significand itself. The first-digit Pearson statistic rises from 8.882 ($p=0.352$) to 117.615, and the continuous CvM statistic rises from 0.078 to 14.676. By contrast, the order-of-magnitude shift leaves both statistics exactly unchanged.

\begin{table}[htbp]
\centering
\caption{Information-boundary illustration. Integer-order magnitude shifts are invisible to log-significand statistics by construction, whereas a small heaping perturbation changes the observed representation. The heaping example is a positive control, not a behavioral model of fraud.}
\label{tab:boundary}
\begin{tabular}{lrrr}
\toprule
Construction & Pearson $X_n^2$ & Pearson $p$ & CvM \\
\midrule
Benford reference & 8.882 & 0.352 & 0.078 \\
Significand-preserving magnitude shift & 8.882 & 0.352 & 0.078 \\
Round-number heaping & 117.615 & $<0.001$ & 14.676 \\
\bottomrule
\end{tabular}
\end{table}

The contrast is deliberately elementary. Digit evidence is informative when competing mechanisms imply different distributions for the observed digit process; it is uninformative about differences that are removed by the map $X\mapsto U$. If legitimate and manipulated processes differ in magnitudes, timing, counterparties, accounting identities, thresholds, or other contextual variables while inducing the same log-significand process, those variables---not a more elaborate digit statistic---must carry the discriminating information.

This also clarifies the role of Benford screening. A digit discrepancy can be useful evidence when a substantive model predicts approximate Benford behavior under legitimate operation and a plausible alternative predicts a departure. The heaping control shows that the information-boundary argument is not a claim that digit analysis has no power. It says instead that such power is alternative-specific. Without substantive competing models, Benford conformity is neither necessary nor sufficient for integrity.

\section{Implications for statistical practice}
\label{sec:practice}

The examples lead to a simple workflow: separate the target distribution, the sampling design, and the substantive interpretation. Table~\ref{tab:practice} summarizes the corresponding actions.

\begin{table}[htbp]
\centering
\small
\caption{Practical interpretation of a Benford discrepancy. The appropriate calibration follows the observation design; the substantive conclusion follows the competing mechanisms.}
\label{tab:practice}
\begin{tabular}{>{\raggedright\arraybackslash}p{0.18\textwidth}>{\raggedright\arraybackslash}p{0.30\textwidth}>{\raggedright\arraybackslash}p{0.39\textwidth}}
\toprule
Setting & What the i.i.d. reference can miss & Practical response \\
\midrule
Serial or ordered data & Within-sequence covariance can inflate or deflate goodness-of-fit statistics. & Preserve order; use sequence/block covariance, a block bootstrap, HAC methods, or simulation from a dynamic null. \\
\addlinespace
Clustered or stratified data & Between-cluster composition and unequal design effects can distort the multinomial reference. & Treat clusters/strata as part of the design; use cluster-level covariance, Rao--Scott-type corrections, or a cluster bootstrap. \\
\addlinespace
Deterministic or quasi-deterministic sequences & A stochastic $\chi^2$ reference may have no natural justification without a randomization mechanism. & State the source of randomness explicitly (for example, phase randomization) or report descriptive discrepancy rather than an automatic $p$-value. \\
\addlinespace
Forensic interpretation & A digit anomaly does not identify its cause, and digit conformity cannot certify features absent from the digit representation. & Specify legitimate and alternative mechanisms and retain contextual variables---magnitudes, timing, counterparties, thresholds, accounting relations---on which they differ. \\
\bottomrule
\end{tabular}
\end{table}

Changing the discrepancy statistic does not solve a sampling problem by itself: the grid-based log-significand CvM experiment shows the same distortion under an i.i.d. reference. Conversely, the benchmark shows that sequence-aware covariance estimation can substantially improve finite-sample calibration when the independent sampling units are correctly identified. The second-order moment-matching correction is more stable than the direct Wald statistic in the smallest-$B$ design studied here; with $B=100$ sequences, the latter rejects 9.4\% of null samples in one design. Large record counts therefore do not substitute for enough independent sampling units.

The practical difficulty is that forensic data do not always arrive with indisputable sequence or cluster boundaries. A transaction file might plausibly be grouped by account, reporting entity, day, counterparty, or administrative unit. If the grouping is misspecified, the estimated covariance need not represent the actual observation law. When natural independent units are not supplied by the design, analysts should justify the grouping substantively and examine sensitivity across plausible clusterings. A model-based dynamic null, HAC/long-run covariance estimate, or block/cluster bootstrap can provide alternative calibrations when their assumptions are more credible. RS2 avoids direct inversion of the full covariance, whereas the Wald procedure requires enough independent units for a stable inverse; neither has a universal minimum-$B$ guarantee.

Diagnosis and attribution remain separate from calibration. Temporal summaries are useful when candidate mechanisms differ in order; conditioning is useful when aggregation is plausible; additional digits or continuous-significand tests are useful only against alternatives that alter the retained significand information. A Benford discrepancy has evidential value when a substantive model explains both why legitimate data should follow the reference and why a relevant alternative should depart from it. The $p$-value alone does not provide that model.

\section{Discussion}
\label{sec:discussion}

Benford's law is a marginal statement; a conventional goodness-of-fit $p$-value is a statement about a sampling model. That distinction is classical in categorical-data and complex-sampling theory, but it is easy to overlook when a familiar digit law is used as an anomaly screen. The controlled constructions here make the distinction unusually transparent: the Benford target is held fixed while persistence, deterministic regularity, and latent clustering alter the sampling distribution of the same discrepancy statistic.

The design-aware experiments add a constructive counterpart. When independent sequence boundaries are observed, the empirical covariance of sequence summaries estimates the weighted-chi-square structure of the Pearson statistic. A second-order Rao--Scott-type moment-matching approximation produces rejection rates close to nominal in the reported benchmark, across wide changes in $B$, $L$, and $n$, and under random mixture composition. The same sampling issue appears for the grid-based log-significand CvM discrepancy. In the wrapped-Gaussian power experiment, correcting size does not eliminate sensitivity to a true marginal departure. These results are not a proposal for a universal Benford test; they illustrate a more general rule that calibration should follow the observation design.

A separate limitation concerns information rather than calibration. If competing mechanisms induce the same complete log-significand process, digit-only statistics cannot distinguish them. The magnitude-shift construction makes that boundary explicit, while the heaping control shows the converse: when an alternative changes the observed significand, digit-based procedures can have substantial power. Forensic interpretation is therefore hypothesis-specific, not a generic consequence of conformity or nonconformity.

\paragraph{Limitations.}
The benchmark is intentionally stylized. The wrapped-Gaussian circular Markov model represents one form of positive circular persistence; the randomized irrational rotation is a deliberately low-discrepancy example rather than a model for deterministic data in general; and the smooth power alternative is only one family of departures. The plug-in moment correction requires correctly identified independent sampling units and a stable estimate of their covariance spectrum. The finite-$B$ experiments show that large record counts do not remove this requirement, especially for the direct Wald statistic. The grid-based CvM calculation is a robustness check rather than a general dependent-data CvM theory. No empirical dataset is used to claim that the four mechanisms describe a particular audit, tax file, scientific dataset, or election.

The conclusion is therefore deliberately narrow: \emph{Benford conformity is a marginal property since a Benford $p$-value is valid only relative to a specified sampling law, statistic, and calibration procedure; and forensic interpretation requires substantive competing models.}

\bibliographystyle{asa-tas}
\bibliography{benford-references}

\appendix

\section{Analytical expectation of the Pearson statistic}
\label{app:pearson}

Let $Y_t^{(d)}=\mathbf 1\{U_t\in A_d\}$ and $N_d=\sum_tY_t^{(d)}$. In each primary construction, $\E(N_d)=np_d$, so
\begin{equation}
\E(X_n^2)
=
\sum_{d=1}^9\frac{\operatorname{Var}(N_d)}{np_d}.
\label{eq:ETgeneral}
\end{equation}
This identity makes explicit that the mean Pearson statistic depends on the joint sampling law even when the marginal probabilities are held fixed.

\subsection{Independent reference}

Under independent sampling,
\[
\operatorname{Var}(N_d)=np_d(1-p_d),
\]
so
\begin{equation}
\E(X_n^2)
=
\sum_{d=1}^9(1-p_d)
=8.
\label{eq:ETiid}
\end{equation}

\subsection{Wrapped-Gaussian multiplicative process}

Write
\[
a_d=\log_{10}d,
\qquad
b_d=\log_{10}(d+1),
\qquad
A_d=[a_d,b_d).
\]
For the digit interval $A_d$, define the Fourier coefficient
\begin{equation}
c_{dk}
=
\int_{a_d}^{b_d}e^{-2\pi iku}\,du
=
\frac{e^{-2\pi i k b_d}-e^{-2\pi i k a_d}}{-2\pi i k},
\qquad k\ne0.
\label{eq:cdk}
\end{equation}
For the wrapped-Gaussian increment kernel in equation~\eqref{eq:rho},
\[
\psi_k
=
\E(e^{2\pi ik\varepsilon})
=
\rho^{k^2}.
\]
Because the wrapped-Gaussian increment distribution is symmetric, each $\psi_k$ is real. Stationarity and Fourier orthogonality give the full lag-$h$ cross-covariance
\begin{equation}
\Gamma_h(d,d')
=
\operatorname{Cov}(Y_t^{(d)},Y_{t+h}^{(d')})
=
\sum_{k\ne0}
c_{d,k}\overline{c_{d',k}}\,\psi_k^h,
\label{eq:gamma-matrix}
\end{equation}
whose diagonal entries are
\begin{equation}
\gamma_d(h)
=
\Gamma_h(d,d)
=
\sum_{k\ne0}|c_{dk}|^2\psi_k^h.
\label{eq:gammad}
\end{equation}
For a sequence of length $L$, the covariance appearing in Proposition~\ref{prop:cluster-clt} is therefore
\begin{equation}
\Omega_L
=
\Gamma_0
+
\sum_{h=1}^{L-1}
\left(1-\frac{h}{L}\right)
\{\Gamma_h+\Gamma_h^\top\}.
\label{eq:omega-longrun}
\end{equation}
This full matrix, not only the cellwise variances, determines the eigenvalue spectrum in equation~\eqref{eq:weightedchisq}. With $B$ independent sequences of length $L$,
\begin{equation}
\operatorname{Var}(N_d)
=
B\left[
Lp_d(1-p_d)
+2\sum_{h=1}^{L-1}(L-h)\gamma_d(h)
\right].
\label{eq:varmarkov}
\end{equation}
Numerical evaluation of the rapidly convergent Fourier sum gives $\E(X_n^2)=11.7106$ for $L=250$ and $\rho=0.5$, compared with the Monte Carlo mean 11.5110. Using the simulated standard deviation 6.620 from Table~\ref{tab:main}, the Monte Carlo standard error of that mean is $6.620/\sqrt{2000}=0.148$; the theory--simulation difference is therefore about 1.35 Monte Carlo standard errors.

\subsection{Irrational rotation}

For the rotation in equation~\eqref{eq:rotation},
\begin{equation}
\gamma_d(h)
=
\lambda\{A_d\cap(A_d-h\theta)\}-p_d^2,
\label{eq:gammarotation}
\end{equation}
where the intersection is taken on the unit circle. Equation~\eqref{eq:varmarkov} then applies with this deterministic covariance. Exact circular-interval intersections yield $\E(X_n^2)=0.2187$, compared with the Monte Carlo mean 0.2228.

\subsection{Balanced latent mixture}

Let
\[
p_{zd}=\Pp(U\in A_d\mid Z=z),
\qquad z\in\{0,1\}.
\]
The balanced design has exactly $B/2$ sequences in each regime and $p_d=(p_{0d}+p_{1d})/2$. Conditional independence within regimes gives
\begin{equation}
\operatorname{Var}(N_d)
=
\frac{BL}{2}\left[
p_{0d}(1-p_{0d})+p_{1d}(1-p_{1d})
\right].
\label{eq:varlatent}
\end{equation}
The resulting analytical expectation is 7.6789, compared with the Monte Carlo mean 7.6765.

\section{Random composition of latent regimes}
\label{app:composition}

Suppose the sequence-level regime is instead drawn independently with probability $1/2$, and condition on the realized fraction $\widehat\pi$ in regime 0. Then
\[
\E(N_d\mid\widehat\pi)
=
n\{\widehat\pi p_{0d}+(1-\widehat\pi)p_{1d}\},
\]
and
\[
\operatorname{Var}(N_d\mid\widehat\pi)
=
n\{\widehat\pi p_{0d}(1-p_{0d})+(1-\widehat\pi)p_{1d}(1-p_{1d})\}.
\]
Decomposing the Pearson expectation into a variance term and a squared displacement from the equal-weight target gives
\begin{align}
\E(X_n^2\mid\widehat\pi)
={}&
\sum_{d=1}^9
\frac{
\widehat\pi p_{0d}(1-p_{0d})+(1-\widehat\pi)p_{1d}(1-p_{1d})
}{p_d}
\notag\\
&+
 n\left(\widehat\pi-\frac12\right)^2
\sum_{d=1}^9\frac{(p_{0d}-p_{1d})^2}{p_d},
\label{eq:composition-exact}
\end{align}
where $p_d=(p_{0d}+p_{1d})/2$. With $a_d=\log_{10}d$ and $b_d=\log_{10}(d+1)$, the two triangular regimes give
\[
p_{0d}=b_d^2-a_d^2,
\qquad
p_{1d}=2(b_d-a_d)-(b_d^2-a_d^2).
\]
Hence
\[
\sum_{d=1}^9\frac{(p_{0d}-p_{1d})^2}{p_d}
\approx 1.28424.
\]
The leading imbalance contribution is therefore equation~\eqref{eq:composition-leading}. Under independent Bernoulli assignment of the $B$ sequence labels, $\E[(\widehat\pi-1/2)^2]=1/(4B)$. Since $n=BL$, the expected imbalance contribution is therefore
\begin{equation}
1.28424\,\frac{L}{4},
\label{eq:composition-unconditional}
\end{equation}
which equals 80.27 for $L=250$. Adding the expected within-regime variance term, 7.6789, gives 87.94, close to the Monte Carlo mean 89.46 from 2,000 random-composition realizations. Equation~\eqref{eq:composition-unconditional} also makes the asymptotic regime explicit. With fixed cluster length $L$ and $B\to\infty$, the additional mean contribution remains of order one rather than growing with the total record count; if $L$ grows, it increases proportionally to $L$. The relevant uncertainty is therefore cluster-level composition, not the number of individual records alone.

\section{A game-theoretic route to the same invariant law}
\label{app:game}

Morrison's multiplication game provides a qualitatively different route to the same Benford/Haar distribution \citep{morrison2010}. On the logarithmic circle, multiplication becomes addition modulo one. For the winning set
\[
W=[0,\log_{10}4),
\]
corresponding to product significands beginning with 1, 2, or 3, Haar randomization by either player fixes the winning probability at $\lambda(W)=\log_{10}4\approx0.60206$ regardless of the opponent's action. Figure~\ref{fig:game-static} illustrates this minimax invariance against four deliberately non-Haar strategies. The Haar row and column stay near the theoretical value, whereas non-Haar pairs can differ substantially.

\begin{figure}[htbp]
\centering
\includegraphics[width=0.78\textwidth]{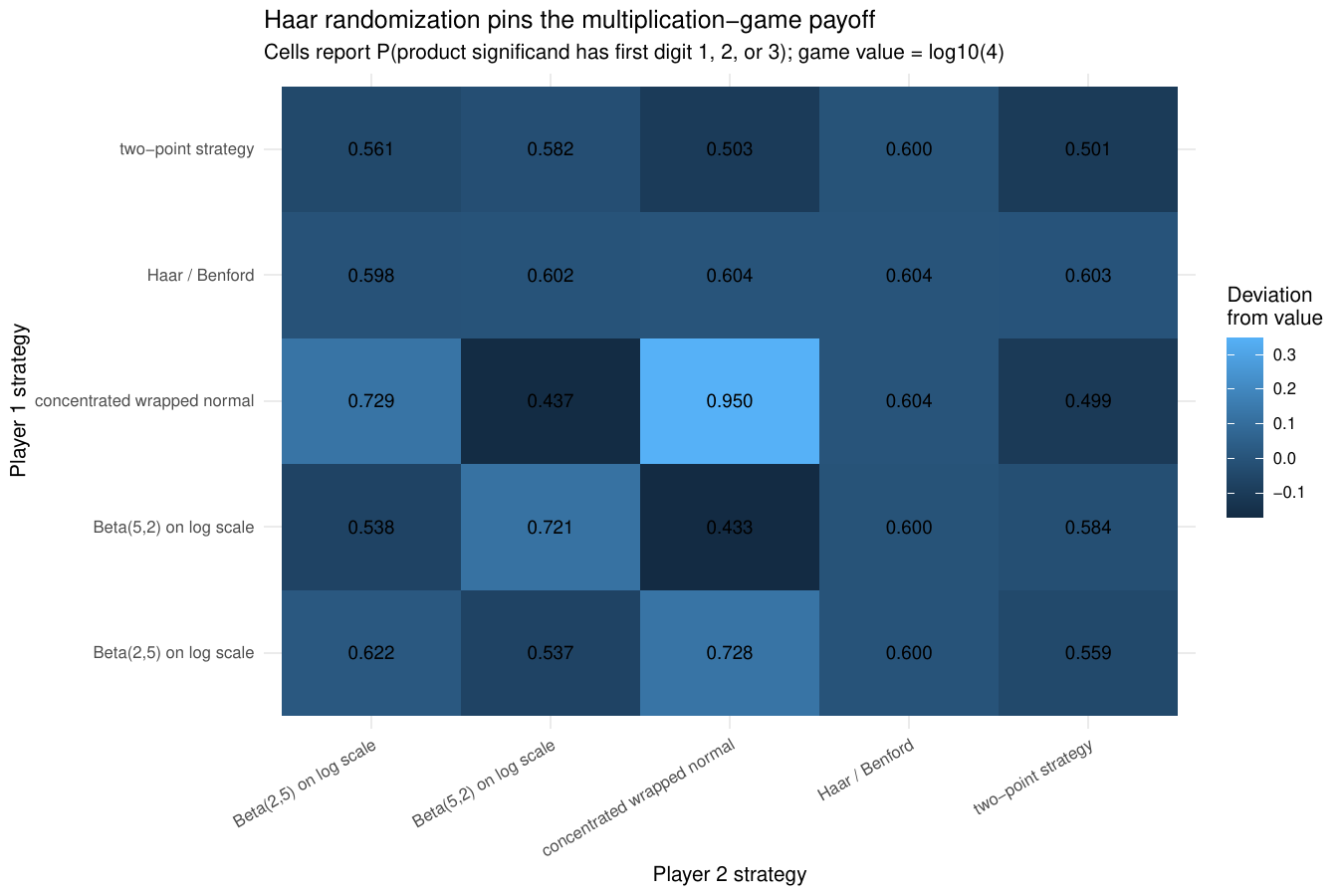}
\caption{Static multiplication game. Haar/Benford randomization fixes the payoff near $\log_{10}4$ regardless of the other player's mixed strategy.}
\label{fig:game-static}
\end{figure}

The reproducible supplement also contains two dynamic illustrations. First, in
\[
U_{t+1}=(U_t+A_t+B_t)\bmod1,
\]
a fresh Haar action by one player makes $U_{t+1}$ uniform conditional on the current state and the opponent's action: translation invariance produces a one-step ``washout.'' Second, if both players use independent wrapped-normal, non-Haar actions, the magnitude of the first circular Fourier coefficient of the state distribution decays geometrically at rate
\[
\eta=\exp\{-2\pi^2(\sigma_A^2+\sigma_B^2)\}.
\]
These examples are included only to emphasize equifinality: strategic invariance and repeated non-Haar interaction can lead to the same marginal law as the nonstrategic constructions in the main text. They are not used in any calibration result.

\section{Additional numerical results}
\label{app:pass3}

\subsection{Exact permutation centering for transition coherence}

For equal-length sequences, let $z_{s,t}=\exp(2\pi iU_{s,t})$. Conditional on the observed values in sequence $s$, the expected complex transition moment after a uniformly random permutation is
\[
\E_{\mathrm{shuffle}}
\left[
\frac{1}{L-1}\sum_{t=1}^{L-1}z_{s,t+1}\overline{z}_{s,t}
\right]
=
\frac{|\sum_{t=1}^Lz_{s,t}|^2-L}{L(L-1)}.
\]
The centered coherence reported in the paper subtracts the average of this exact sequence-level baseline before taking the modulus. This removes coherence attributable only to the composition of a sequence while retaining ordered transition structure.

\subsection{Full \texorpdfstring{$B$, $L$, and $n$}{B, L, and n} sensitivity grid}

Table~\ref{tab:design-size-full} reports all designs summarized in Table~\ref{tab:design-size-summary}. The duplicated $(B,L)=(400,250)$ design appears in two independently seeded grids and consequently differs by ordinary Monte Carlo error.

\begin{table}[htbp]
\centering
\small
\caption{Full design-size sensitivity grid for the wrapped-Gaussian circular Markov construction at $\rho=0.5$ (500 replications per row).}
\label{tab:design-size-full}
\begin{tabular}{lrrrrrr}
\toprule
$(B,L)$ & $n$ & Pearson & RS1 & RS2 & Wald & Mean $\lambda$ \\
\midrule
$(400,10)$ & 4,000 & 0.196 & 0.046 & 0.038 & 0.072 & 1.375 \\
$(400,25)$ & 10,000 & 0.250 & 0.054 & 0.048 & 0.064 & 1.432 \\
$(400,50)$ & 20,000 & 0.202 & 0.060 & 0.054 & 0.062 & 1.450 \\
$(400,100)$ & 40,000 & 0.268 & 0.082 & 0.064 & 0.056 & 1.456 \\
$(400,250)$ & 100,000 & 0.212 & 0.078 & 0.058 & 0.058 & 1.462 \\
\addlinespace
$(100,1000)$ & 100,000 & 0.234 & 0.062 & 0.046 & 0.094 & 1.469 \\
$(200,500)$ & 100,000 & 0.212 & 0.052 & 0.042 & 0.060 & 1.468 \\
$(400,250)$ & 100,000 & 0.234 & 0.064 & 0.048 & 0.054 & 1.461 \\
$(1000,100)$ & 100,000 & 0.210 & 0.054 & 0.048 & 0.040 & 1.461 \\
$(2000,50)$ & 100,000 & 0.204 & 0.062 & 0.042 & 0.044 & 1.450 \\
\bottomrule
\end{tabular}
\end{table}

\subsection{Continuous-log-significand calibration and power}

The grid-based CvM calculation in equation~\eqref{eq:cvm} is used only to check that the sampling-design distinction is not a consequence of first-digit binning. Table~\ref{tab:cvm-calibration-app} reports its null calibration.

\begin{table}[htbp]
\centering
\caption{Nominal 5\% rejection rates for the continuous-log-significand Cram\'er--von Mises statistic under the conventional i.i.d. reference and a sequence-aware second-order covariance calibration (500 replications).}
\label{tab:cvm-calibration-app}
\begin{tabular}{lrr}
\toprule
Construction & i.i.d. reference & Design correction \\
\midrule
i.i.d. Benford & 0.056 & 0.062 \\
Stochastic multiplicative & 0.222 & 0.048 \\
Randomized irrational rotation & 0.000 & 0.044 \\
Balanced latent mixture & 0.030 & 0.064 \\
\bottomrule
\end{tabular}
\end{table}

Table~\ref{tab:power-full} reports the complete power grid used for Figure~\ref{fig:power}.

\begin{table}[htbp]
\centering
\caption{Rejection probabilities under the smooth marginal alternative $f_\epsilon(u)=1+\epsilon(2u-1)$ with the same latent $\rho=0.5$ circular Markov construction (500 replications per row).}
\label{tab:power-full}
\begin{tabular}{rrrrrr}
\toprule
$\epsilon$ & Pearson i.i.d. & Pearson RS2 & Pearson Wald & CvM i.i.d. & CvM design \\
\midrule
0.000 & 0.246 & 0.064 & 0.074 & 0.250 & 0.060 \\
0.005 & 0.278 & 0.072 & 0.072 & 0.332 & 0.092 \\
0.010 & 0.406 & 0.118 & 0.122 & 0.542 & 0.210 \\
0.020 & 0.734 & 0.438 & 0.412 & 0.878 & 0.628 \\
0.030 & 0.964 & 0.830 & 0.796 & 0.994 & 0.934 \\
0.050 & 1.000 & 1.000 & 1.000 & 1.000 & 1.000 \\
\bottomrule
\end{tabular}
\end{table}

\subsection{Other diagnostics}

The reproducible companion also reports the Fourier discrepancy
\[
B_5
=
\sum_{k=1}^5
\left|
\frac1n\sum_{j=1}^n e^{2\pi ikU_j}
\right|^2,
\]
together with pairwise rank-discrimination summaries and sensitivity of $C_{\mathrm{centered}}$ to $\rho$. These diagnostics are not used to calibrate any generator. They support the same distinction made in the main text: empirical log-significand statistics can discriminate sampling laws even when the population marginal is identical, while order-sensitive diagnostics directly recover structure that marginal summaries discard.

Figure~\ref{fig:rho-coherence} shows that the centered transition coherence tracks the first circular Fourier coefficient over the sensitivity grid, with the largest finite-block deviation at high persistence.

\begin{figure}[htbp]
\centering
\includegraphics[width=0.68\textwidth]{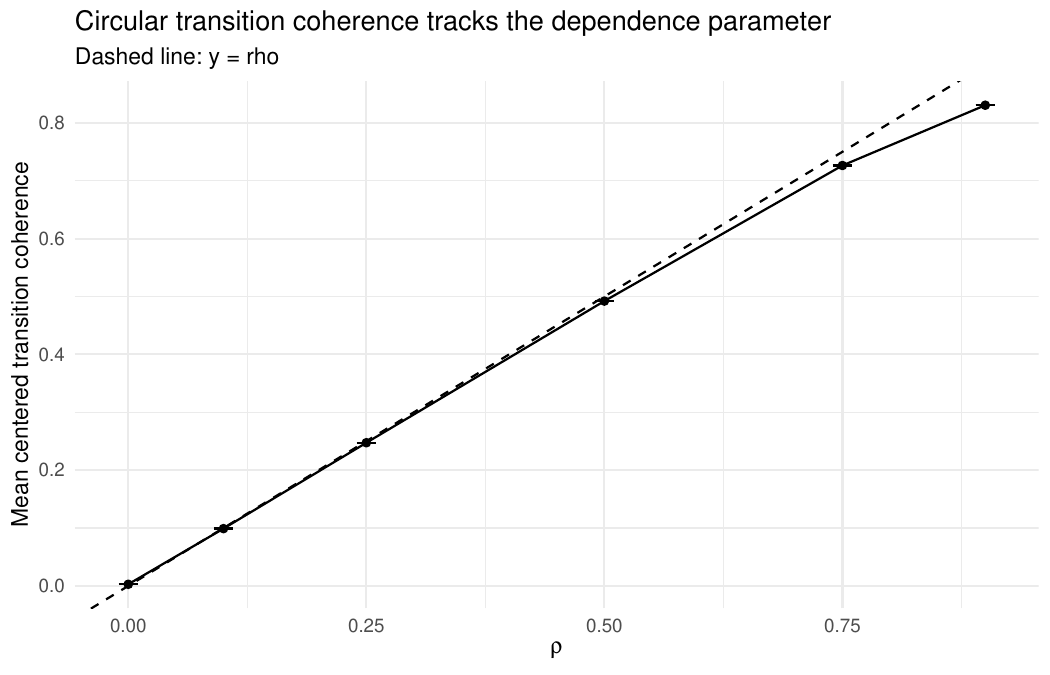}
\caption{Centered circular transition coherence across the wrapped-Gaussian persistence grid. The dashed line is $y=\rho$.}
\label{fig:rho-coherence}
\end{figure}

\end{document}